\documentclass[sigconf]{cidr-2024}

\usepackage{graphicx}
\usepackage[utf8]{inputenc}
\DeclareUnicodeCharacter{2212}{-}
\usepackage{graphicx}
\usepackage{mdwlist}
\usepackage{subcaption}
\usepackage[textsize=small]{todonotes}
\usepackage{listings}
\usepackage{balance}
\usepackage{xspace}
\usepackage{caption}
\usepackage{hyperref}
\usepackage{soul}
\usepackage{float}
\usepackage{enumitem}
\usepackage{multirow}
\usepackage{algorithm}
\usepackage{algpseudocodex}
\usepackage{microtype}
\usepackage[subtle]{savetrees}
\usepackage[normalem]{ulem}
\usepackage{amsmath,amsfonts,amsthm}
\usepackage{mathtools}
\usepackage[toc,title]{appendix}
\usepackage{asymptote}
\usepackage{geometry}
\usepackage{ocgx2}
\usepackage[absolute,overlay]{textpos}

\theoremstyle{plain}
\newtheorem{theorem}{Theorem}
\newtheorem{proposition}[theorem]{Proposition}

\theoremstyle{remark}
\newtheorem{remark}{Remark}

\usepackage[nameinlink,capitalise]{cleveref}
\crefformat{equation}{(#2#1#3)}  

\newcommand{\vect}[1]{\boldsymbol{#1}}  
\newcommand{\sgn}{\operatorname{sgn}}   

\newcommand{\sys}{cloud oracle\xspace}

\newcommand{\placement}{placement decision\xspace}

\newcommand{\placements}{placement decisions\xspace}

\newcommand{\placementS}{decision\xspace}
\newcommand{\placementsS}{decisions\xspace}

\newcommand{\insight}{\textit{\textbf{Implications}:\xspace}}

\newcommand{\shorten}[1]{\xspace}
    
\DeclareMathOperator*{\argmin}{arg\,min}

\definecolor{code_indent}{HTML}{CCCCCC}

\AtBeginDocument{%
  \providecommand\BibTeX{{%
    \normalfont B\kern-0.5em{\scshape i\kern-0.25em b}\kern-0.8em\TeX}}}

\acmConference[CIDR '24]{Make sure to enter the correct
  conference title from your rights confirmation emai}{January 14--17,
  2024}{Chaminade, USA}

\begin{document}

\title[Optimizing the cloud? Don’t train models. Build oracles!]{Optimizing the cloud? Don’t train models. Build oracles!}
\subtitle{
\begin{ocg}[printocg=always,exportocg=never]{Print Only}{fig:printonly}{false}
\small{You have unlocked the cloud oracle icon by printing this paper. We hope you like the icon and this paper. Let us know when our paths cross.}
\end{ocg}
}
\begin{textblock*}{1pt}(20pt,50pt) 
\begin{ocg}[printocg=always,exportocg=never]{Print Only}{fig:printonly}{false}
  \includegraphics[width=2cm]{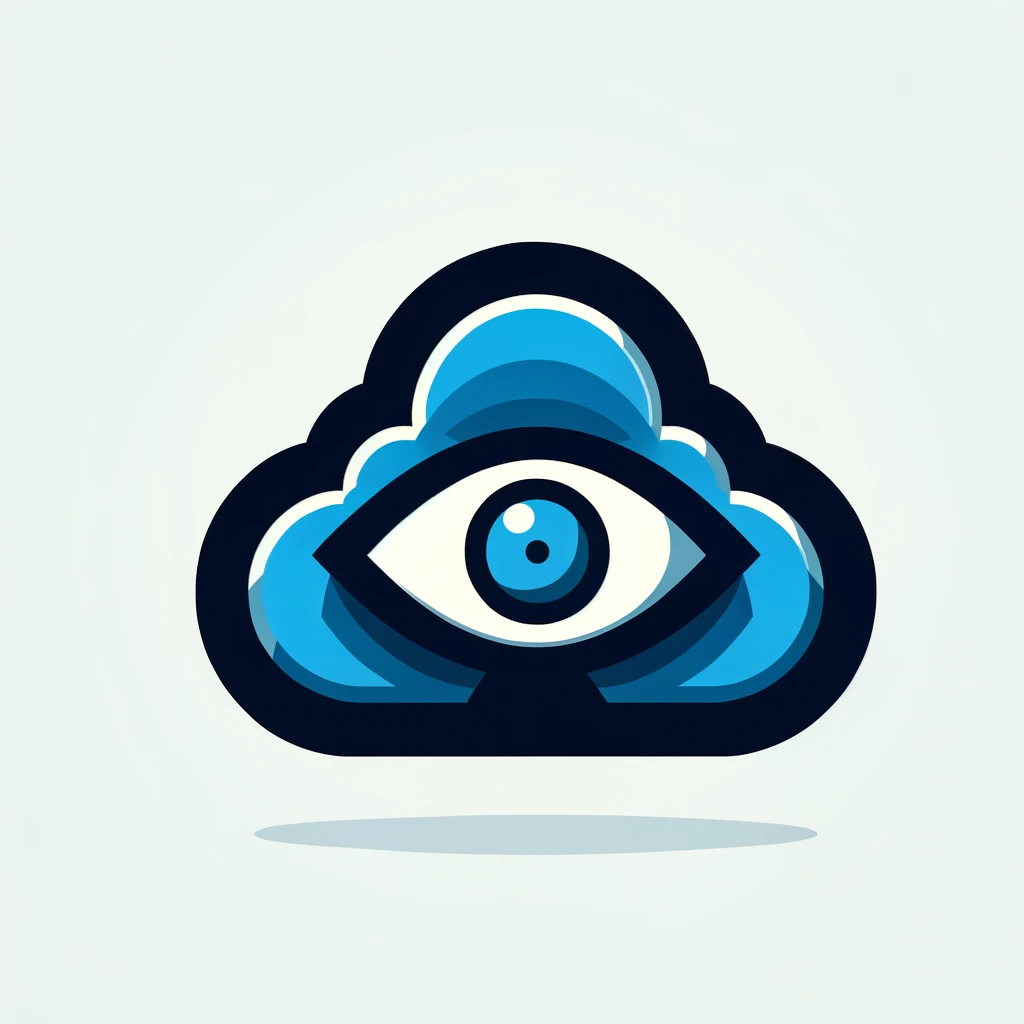}
\end{ocg}
\end{textblock*}

\author{Tiemo Bang, Conor Power, Siavash Ameli, Natacha Crooks, Joseph M. Hellerstein}
\affiliation{%
  \institution{University of California, Berkeley}
  \country{}
}
\email{tbang@berkeley.edu}

\renewcommand{\shortauthors}{Bang, et al.}

\begin{abstract}
We propose \textit{cloud oracles} as an alternative to machine learning for online optimization of cloud configurations. Our cloud oracle approach guarantees complete accuracy and explainability of decisions for problems that can be formulated as parametric convex optimizations. We give experimental evidence of this technique's efficacy and share a vision of research directions for expanding its applicability.
\end{abstract}

\maketitle

\section{Introduction}
\label{sec:intro}

\begin{figure}
    \centering
    \includegraphics[height=4cm]{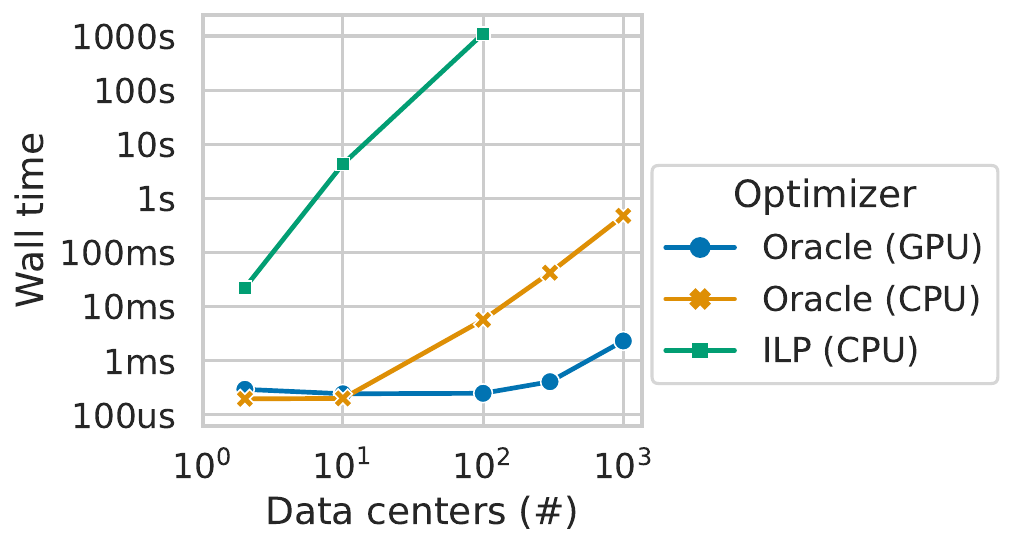}
    \caption{Optimization of object storage location for minimal access latency among $n$ data centers---via the traditional exact approach (ILP) and our exact \sys{}.}
    \label{fig:headline}
\end{figure}

The cloud offers the promise of an infinitely flexible, elastic, and scalable black box. In practice, to keep costs low, developers must intelligently choose how to deploy their applications between hundreds of storage tiers, hardware types, and regions. This is exacerbated by the move to multi-cloud pipelines and the proliferation of cloud providers~\cite{jain_skyplane_2023, yang_skypilot_2023, chasins_sky_2022}. Finding the optimal configuration for a system is far from straightforward. This challenge is further compounded by the fact that one must find this configuration not just once, but multiple times as workload changes, which happens frequently at scale.

Existing solutions cluster around two extremes. 
On the one hand, Integer Linear Programming (ILP) solvers offer provably correct solutions but at prohibitively high computational cost.
For instance, our benchmarks highlighted in Figure~\ref{fig:headline} show that optimizing the storage location of a single object takes >16 minutes at the current cloud-scale (>100 data centers)~\cite{cloudflare_datacenters,aws_obj,gcp_obj,azure_obj}.
On the other hand are approximation solutions, currently dominated by machine learning techniques. These trade some accuracy for fast execution but suffer from operational challenges of training pipelines and a lack of explainability that hinders deployment in production settings~\cite{autonomous_azure, kea}.
This accuracy trade-off is also an expensive one at cloud scales where companies are paying hundreds of millions in cloud bills annually. 
Recent work on ML-based autotuning for cloud databases~\cite{onlinetune} showed performance variations in the 10s of percents (millions of dollars of waste, at the scale discussed above), as well as outright system crashes due to unsafe configuration recommendations.

In this paper, we propose a third avenue: cloud oracles. We observe that many decisions in the cloud are (yet another) scenario that can be viewed as a query optimization problem! This perspective on cloud decisions leads us to a little-used hammer in the query optimization toolbox: parametric query optimization (PQO).

Like machine learning, PQO splits optimization into an offline precomputation phase and an online serving phase. At compile time, PQO, based on a known cost model, materializes several execution plans, each one of which is optimal for a subset of all possible input parameters.
At runtime, when the input parameters become known, PQO can then select the final, correct plan with near-zero overhead, as all relevant optimal plans have already been indexed~\cite{ioannidis_parametric_1997}.
PQO's success hinges on \textbf{three conditions}:
1) a relatively fixed cost model
for long-lasting offline decisions;
2) a ``well-behaved'' cost model for tractable offline computation of all optimal decisions;
3) a tractable decision space---either enumerable in reasonable time, or collapsible/prunable into a tractable number of equivalence classes.
\shorten{\footnote{To use the classic example of join enumeration, there are $\frac{(2(n-1))!} {n!(n-1)!)}$ binary join trees of $n$ leaves (the $n-1$’st Catalan number), but thanks to the principle of optimality only $(n-1)2^{n-2}$ feasible joins to explore in dynamic programming [OnoLohmanVLDB90]. For $n=20$ this is $1.76726319\times 10^9$ vs. $4.980736\times 10^6$.}}

Cloud optimization broadly satisfies these conditions.
Pricing and SLOs are reliable and fairly slow-changing by necessity, to maintain stable business relations and service operation. Also, the pay-per-use cloud billing model translates to simple linear cost functions: $0.3$ per read and $0.5$ per write for instance.
Finally, the unbounded elasticity of cloud infrastructure prevents ``bin-packing-style'' pitfalls where due to limited capacity one partial decision changes the options available for the next partial decision---impeding pruning.
As a result, the analogy from PQO to cloud oracles is quite direct:
query plans map to system configurations, e.g., choices of VM or storage resources across regions or tiers; query execution costs map to price- and latency-based cost models; workload parameters like query selectivity map to API request metrics like \texttt{PUT} and \texttt{GET} rate.

Given the fixed cost model and the space of possible parameters, we propose to precompute a data structure that can be leveraged online
to frequently compute the optimal cloud configuration for the current workload.
To that end, it is natural in today's environment
to consider an ML approach, and train a probabilistic or heuristic model for this task. However, 
given well-behaved cost functions and the financial penalties for misprediction in this setting, we propose instead to adopt an exact approach from computational geometry (\S\ref{sec:motivation:us}):
a convex polytope that represents all feasible and non-dominated configurations.
This polytope is both compact (on the order of GBs) and deterministically optimal.

Calculating the optimal configuration for a given workload corresponds to finding the lowest point on the perimeter of the polytope matching each parameter.
Conveniently, this can be computed quickly and with exact accuracy via a simple geometric calculation based on ray-shooting, described in \S\ref{sec:motivation:us}.
As the workload shifts, we can use the precomputed polytope and the fast online ray-shooting calculation to recompute the solution for the new input.
Figure~\ref{fig:headline} highlights that the \sys{} returns these optimization results within milliseconds rather than minutes---orders of magnitude speedup without loss of accuracy.

We refer to our approach as the \textit{cloud oracle} approach because at runtime we are able to quickly extract exact answers from the data structure (oracle) we built in the offline step.
Source code is available on GitHub: \url{https://github.com/hydro-project/cloud_oracle}.
\section{From Parametric Optimization to Cloud Oracles}
\label{sec:motivation}
In the cloud, the design space is massive, we can reconfigure at will, and the stakes are high. Under these circumstances, there are two properties we want from a cloud optimization technique: \textit{accuracy} and \textit{velocity}. Accuracy refers to the distance the computed configuration is from the true optimal configuration. Velocity refers to the frequency with which we can recompute the optimal configuration online. 
In this section, we describe the two current approaches to optimization in the cloud: ILPs (high accuracy, low velocity) and ML (lower accuracy, high velocity). We argue for a third approach: explicit cloud oracles based on parametric optimization. This approach offers the best of both worlds for any problem amenable to a parametric representation.

\subsection{The Rising Star: Machine Learning}

The use of ML for decision-making in data systems has rapidly gone from a vision~\cite{seattle_report, sageDB} to a production reality~\cite{cosmos_big_data, redshift_reinvented, kea, autonomous_azure}. Fundamentally, machine learning does heavyweight offline computation (training) in order to achieve high-velocity online decision-making (inference). Its use cases in systems primarily replace the greedy or heuristic algorithms that were developed in prior decades.
In many scenarios, we see big accuracy improvements from machine learning over previous heuristic-based approaches to approximation~\cite{ready_for_cardinality}. These solutions still fall short of perfect accuracy though and their improved accuracy over classical heuristics comes at a heavy price: computationally heavy training and lack of explainability.

First, real-world training of machine learning is highly burdensome. The process of gathering production telemetry at scale, testing and verifying these models, and preventing performance regressions has fostered the growth of an entirely new field of research and industry. Production \textit{ML ops} is highly complex~\cite{shreya_rolando, cloudy_with_dbms} and with increasingly stringent privacy regulations, it is only getting more challenging~\cite{private_ml}.

Second, the lack of predictability and interpretability of decisions is a major barrier to production adoption of machine learning techniques. At Microsoft, researchers describe the production team's preference for heuristics over machine learning and for simple linear or tree models over the more complex models used in research settings~\cite{autonomous_azure, kea}. These preferences are due to the uninterpretability of models like deep learning models, which are often the models that outperform traditional solutions in research studies~\cite{ready_for_cardinality}.

For problems that can be expressed analytically, parametric query optimization offers a workaround to the headaches of ML ops and gives exact and interpretable answers for all decisions. To understand the parametric approach, we must first understand the linear programming techniques that it is built on.

\subsection{The Classicist: (Integer) Linear Programming}

Linear programming (LP) is a mathematical optimization technique that has been in use since the mid-20\textsuperscript{th} century. LP finds the input that minimizes some linear objective function, where the valid inputs are constrained by a set of linear equations. Many decision problems can be formulated in this way including data partitioning~\cite{matei_partitioning}, view materialization~\cite{cloudviews}, checkpoint selection~\cite{phoebe}, and resource allocation~\cite{pop}. While continuous linear programming problems can be solved with a runtime complexity of approximately $O(n^3)$~\cite{bewley_complexity_1987}, the integer versions of the problem are exponentially hard~\cite{karp1972reducibility}. In Integer Linear Programming (ILP), only inputs with integer values are valid. Similarly, Mixed Integer Linear Programming (MILP) requires that some but not all of the inputs be integers. Both ILP and MILP problems take exponential time to solve---they are NP-hard~\cite{karp1972reducibility}.

Intuitively, ILP is harder than LP as the search space in ILP consists only of discrete points whereas LP considers smooth surfaces that can be walked, and whose intersections can be computed cheaply. In particular, the valid solutions to LP problems form a convex polytope. This polytope is the geometric shape enclosed  by all the hyper planes (linear constraints) of the problem. Finding the intersection points on the perimeter of this polytope is sufficient for computing the optimal solution. Throughout this paper we will see that having this polytope representation of optimal solutions is very powerful, as we can use it not just to compute the answer on our current inputs efficiently, but to do what-if analysis and drift analysis efficiently as well (\S\ref{sec:use-cases}).

\subsection{The New Kid: Explicit Cloud Oracles}
\label{sec:motivation:us}

In this work, we observe that parametric query optimization
~\cite{trummer_multi-objective_2017,doshi_kepler_2023,bruno_configuration-parametric_2008,vaidya_leveraging_2021,ioannidis_parametric_1997} and cloud optimization are extremely similar. As such, many of the ideas of PQO can be repurposed for efficient optimization in the cloud context. PQO not only has the potential to achieve both the accuracy of ILPs and the velocity of ML, but it also enables, as we describe next, highly expressive optimization tasks.

``Parametric query optimization attempts to identify at compile time several execution plans, each one of which is optimal for a subset of all possible values of the run-time parameters. At runtime, when the actual parameter values become known, it is then possible to choose the final optimal plan with essentially no overhead.''~\cite{ioannidis_parametric_1997}
We propose transferring this idea to cloud optimization in the form of \textit{\sys{}s}---a precomputed data structure of optimization decisions, akin to a multidimensional index of the parameter space, which enables lightning-fast online optimization queries.

\emph{Offline Optimization of a Cloud Oracle.}
As a first step in alleviating optimization overhead, we advocate for offline precomputation of \sys{}s based on PQO techniques.
This is loosely analogous to the computation involved in training ML models, without any of the ML burdens of training data, model architecture selection, or hyperparameter tuning.
PQO essentially comprises 1) \textit{plan enumeration} in a fairly static problem setting and 2) a \textit{linear cost model} allowing for efficient comparison of plans.
Decisions in cloud optimization problems are equally enumerable and also have linear cost models.
We can thus build a \sys{} as the exhaustive collection of all optimal decisions over all possible parameters---enumerating and judging decisions similar to parametric query optimization.

\emph{Online Optimization by Geometry.}
Once the \sys{} has been precomputed offline, the next step is querying the oracle in a way that is both fast and accurate.
We can frame optimization tasks on a \sys{} as computational geometry queries on a convex polytope.
Rather than indexing the decisions, we point out that the exhaustive set of their linear cost functions forms a single convex polytope.
We find that we can efficiently query the convex polytope even when the \sys{} comprises millions of \placementsS with high-dimensional cost functions.
This is in part thanks to the ability to reuse GPU hardware and linear algebra packages commoditized by ML. These can be used to perform highly optimized computational geometry queries on this convex polytope.
Moreover, having access to the exhaustive set of all optimal decisions allows us to venture out to further optimization tasks beyond plain minimization tasks.

More formally, we represent the convex polytope as a dense matrix of hyperplanes derived from the cost functions of the decisions contained in the \sys{}.
That is, the linear cost function of a \placementS has the type $f: \mathbb{R}^n \rightarrow \mathbb{R}$, which can also be viewed as a hyperplane in $\mathbb{R}^{n+1}$ space---the $parameter \times cost$ space of the parameters and the resulting costs.
The space enclosed by all these hyperplanes is the convex polytope and the surface of this convex polytope represents the lowest costs for all parameters.

In \S\ref{sec:vision} we discuss the set of problems currently amenable to such \sys{} constructions and the open research avenues to expanding that set of problems. But first, let us jump into a concrete example.
\vspace{-2em}
\section{Case Study: Fault-tolerant Object Placement}
\label{sec:use-cases}

To make things concrete,
we explore the situation of building a \sys{} for fault-tolerant object storage
and walk through three online decision scenarios:
\textbf{which} configuration is optimal now,
\textbf{what if} we were to make a large change to the workload, and
\textbf{when} will we need to reconfigure in the future.

For our case study, consider an online retailer who uses object storage to manage its warehouse inventory---storing for each inventory item in each warehouse one object in Cloudflare's R2 object store~\cite{cloudflare_obj}.
The retailer wants to keep operating even in the face of disaster, e.g., power outage, but seeks to minimize access latency for their globally distributed clients.
Specifically, the retailer wants to store two up-to-date copies of each object in two of Cloudflare's 300 data centers~\cite{cloudflare_datacenters}
that are at least 200km apart.
Clients are permitted to read any one of the copies but must update both, so that the read latency is the minimum of the network latency between a client and the two data centers while the write latency is the maximum.

We define the fault-tolerant object placement problem for an object $o$ as follows.
Let $c \in C$ be the client data centers from which object $o$ is accessed.
Let $d \in D$ be the data centers that may host object $o$.
Let $\vec{l}$ be the matrix with the network latency between clients and data centers. Let $\vec{w}$ and $\vec{r}$ be the workload vectors containing the write frequencies and read frequencies of all clients.
If object $o$ is hosted in data centers $d$ and $d'$, the overall access latency is calculated as in Equation~\ref{eq:cost_function}.
The objective is to minimize overall access latency while choosing two data centers with $\geq$200km distance:

\begin{align}
    \forall d,d' \in D: f_{d,d'}:= \mathbb{R}^{|C|} \times \mathbb{R}^{|C|} \rightarrow \mathbb{R} \\
    f_{d,d'}(\vec{w}, \vec{r}) = \sum_{c \in C} max(\vec{l}_{c,d},\vec{l}_{c,d'})\vec{w}_c + min(\vec{l}_{c,d},\vec{l}_{c,d'})\vec{r}_c \label{eq:cost_function} \\
    \argmin_{d, d' \in D}\quad f_{d,d'}(\vec{w}, \vec{r}) \label{eq:argmin}\\
    s.t.\quad dist(d, d') \geq 200 \label{eq:constraint}
\end{align}

\begin{algorithm}[t]
\begin{align}
    &min \sum_{c \in C} z_c + \sum_{d \in D} l_{c,d} r_c y_{c,d}\text{ s.t.} &\text{\small{\#minimize total latency}}\\
    & \sum_{d \in D}x_d = 2, \forall c \in C & \text{\small{\#write 2 DCs}}\\
    & \sum_{d \in D} y_{c,d} = 1, \forall c \in C &\text{\small{\#read 1 DC}}\\
    &x_d \geq y_{c,d}, \forall c \in C, \forall d \in D &\text{\small{\#read only if writing}}\\
    &x_d, y_{c,d} \in \{0,1\}, \forall c \in C, \forall d \in D &\text{\small{\# no/yes decisions}}\\
    &z_c \geq l_{c,d} r_c x_d, \forall d \in D, \forall c \in C & \text{\small{\#aux. write lat. $max(l_{c,d}r_cx_d)$}}\label{eq:write_aux:start}\\
    &z_c \in \mathbb{R}, \forall c \in C &\text{\small{\#aux. var for write latency}}\label{eq:write_aux:end}\\
    &200 v_{d,d'} \leq dist_{d,d'}, \forall d,d' \in D & \text{\small{\#linear dist. constraint}}\label{eq:dist_constraint:start} \\
    &2v_{d,d'} \leq x_d + x_{d'}, \forall d,d' \in D & \text{\small{\#aux. mapping of $v$ to $x$:}}\\
    &v_{d,d'} \geq x_d + x_{d'} - 1, \forall d,d' \in D & \text{\#\small{$v_{d,d'} \iff x_d \wedge x_{d'}$}}\\
    &v_{d,d'} \in \{0,1\}, \forall d,d' \in D &\text{\small{\#aux. for dist. constraint}}\label{eq:dist_constraint:end}
\end{align}
\caption{ILP formulation for the optimal decisions of the fault-tolerant object placement problem with distance constraint. The distance constraint is linearized via auxiliary variables and constraints, as its straightforward formulation would be quadratic.}
\label{lst:ILP}
\end{algorithm}

\paragraph{The Classic ILP Formulation:}
Traditionally, one would encode this problem into the ILP in Algorithm~\ref{lst:ILP}. If you're not familiar with ILP formulations it is okay to ignore this diagram. 

Overall, 0-1 variables $x_d$ encode the binary decision of whether to store an object in a specific data center $d$
and $y_{c,d}$ encodes the decision of which client reads which copy.
The cost of these decisions in terms of write/read latency for the given write/read frequencies ($\vec{w}$, $\vec{r}$) and network latencies $\vec{l}$ is encoded in the linear coefficients $l_{c,d}$.
However, the \emph{max} terms and distance constraints do not permit straightforward linear encode. Their linear formulation requires auxiliary variables and constraints---Eq.~\ref{eq:write_aux:start}--\ref{eq:write_aux:end} and Eq.~\ref{eq:dist_constraint:start}--\ref{eq:dist_constraint:end}, respectively.
As a result, the ILP grows large and complex even for relatively few data centers, and we will see its prohibitive overhead in Section \ref{sec:use-cases-minimization}.

\subsection{Offline Computation of the Cloud Oracle}

Fortunately, we can compute a \sys{} for the fault-tolerant object placement problem parameterized on the read and write frequencies (the workload which changes over time). The \sys{} computation is possible because our three conditions all hold:
\begin{itemize}[noitemsep,nosep,left= 0pt .. 8pt]
    \item 
\emph{Condition 1--Fixed cost model}:
Cloud vendors like Cloudflare offer reliable network latency~\cite{cloudflare_latency_1,cloudflare_latency_2,cloudflare_latency_3}.
In practice, we hence can consider the network latencies $\vec{l}$ as static coefficients rather than parameters in our cost function (Eq.~\ref{eq:cost_function}).

\item \emph{Condition 2--Tractable cost model}: Our offline computation can easily resolve the \emph{max} and \emph{min} terms to plain linear coefficients (explained in Algorithm \ref{lst:construction}). This makes the cost model linear and tractable.

\item \emph{Condition 3--Tractable decision space}:
Object stores promise elastic capacity for reasonable workloads~\cite{cloudflare_obj,aws_obj,gcp_obj,azure_obj}.
We can assume that objects do not compete for capacity, so that the decisions space collapses to the placement of a single object.
Also,  we can assume that the optimal data centers always have capacity, so that we can prune out non-optimal decisions without regret.
\end{itemize}

Algorithm~\ref{lst:construction} implements the offline computation of a \sys{} for given client/storage data centers ($C$, $D$) and network latencies between those ($\vec{l}$). At a high level, we enumerate the data center pairs that satisfy the fault tolerance requirements and compute their read and write latencies.
Then we prune out strictly non-optimal pairs and add the remaining pairs to a matrix representation. More specifically:
\begin{itemize}[noitemsep,nosep,left= 0pt .. 8pt]
    \item
    \emph{Lines~\ref{line:construction:enum:start}--\ref{line:construction:enum:end}}:
    We enumerate the entire search space, apply the constraint (Eq.~\ref{eq:constraint}), and importantly compute the \emph{min} and \emph{max} terms of the objective function to get the concise network latency $\vec{l}_p$ of each valid placement $p$.
    \item
    \emph{Lines~\ref{line:construction:filter:start}--\ref{line:construction:filter:end}}:
    We now filter out irrelevant \placements. Without knowing the actual write and read frequencies ($\vec{w}$, $\vec{r}$),
    we can clearly tell that a \placement never yields lowest total access latency when its latency coefficients are dominated ($\vec{l}_p > \vec{l}_{p'}$).
    The result is a compact yet accurate set of \placements that are optimal for some write/read frequencies.
    \item \emph{Lines~\ref{line:construction:transformation:start}--\ref{line:construction:transformation:end}}:
    We finally construct the \sys{} as a matrix containing the latencies of all optimal \placements.
    We reformulate the linear coefficients $\vec{l}_p$ as hyperplanes in $parameter \times latency$ space, adding an additional $-1$ coefficient.
    This is the textbook reformulation: $\vec{l}_p \cdot \vec{a} \Leftrightarrow x =\vec{l}_p \cdot \vec{a} \Leftrightarrow 0 =\vec{l}_p \cdot \vec{a} - x$, where $\vec{a}=[\vec{w}_0,\dots,\vec{w}_{|C|},\vec{r}_0,\dots,\vec{r}_{|C|}]$.
    We now have a dense linear representation of all optimal \placements, allowing fast and accurate online optimization.
\end{itemize}

\begin{algorithm}[t]
\caption{Algorithm for computing a \sys{} for fault-tolerant object placement problem, given data centers $D$ that may store objects, client data centers $C$ that access objects, and the network latency $\vec{l}$ between those data centers.}
\label{lst:construction}
\begin{algorithmic}[1]
\Procedure{compute\_oracle}{$C$, $D$, $\vec{l}$}
  \State $placements \gets \{\;\}$ \Comment{Initialize decisions of oracle}
  \LComment{Enumerate valid data center pairs as potential placements}
  \ForAll{$d, d' \in D$}\label{line:construction:enum:start}
    \If{$dist(d,d') \geq 200$} \Comment{Check distance constraint}
    \State $p \gets (d,d')$ \Comment{Valid placement p}
    \LComment{Write/read latencies of p by element-wise max/min}
    \State $\vec{l}_p \gets [max(l_{0,d},l_{0,d'}),\dots, max(l_{|C|,d},l_{|C|,d'}),$
    \State $\quad \quad min(l_{0,d},l_{0,d'}),\dots, min(l_{|C|,d},l_{|C|,d'})]$ 
    \State $placements \gets placements \bigcup (p, \vec{l}_p)$
    \EndIf
  \EndFor\label{line:construction:enum:end}
  \LComment{Filter out placements with strictly higher latency}
  \ForAll{$(p,\vec{l}_p), (p',\vec{l}_{p'}) \in placements$}\label{line:construction:filter:start}
  \If{$ \vec{l}_{p'} < \vec{l}_p$}
  \State $placements \gets placements \setminus \{(p, \vec{l}_p)\}$
  \State \textbf{break}
  \EndIf
  \EndFor \label{line:construction:filter:end}
  \LComment{Convert into matrix of planes with coeff.\; $\vec{l}_p$ and extra -1}
  \State $planes \gets \begin{bmatrix}
\vec{l}_{0,p} & \cdots & \vec{l}_{2|C|,p} & -1 \\
\vdots
\end{bmatrix}, \forall (p, \vec{l}_p) \in placements$\label{line:construction:transformation:start}
\State
  \State \Return Oracle($placements$, $planes$)\label{line:construction:transformation:end}
\EndProcedure
\end{algorithmic}
\end{algorithm}

\begin{figure}
    \centering
    \includegraphics[height=4.5cm]{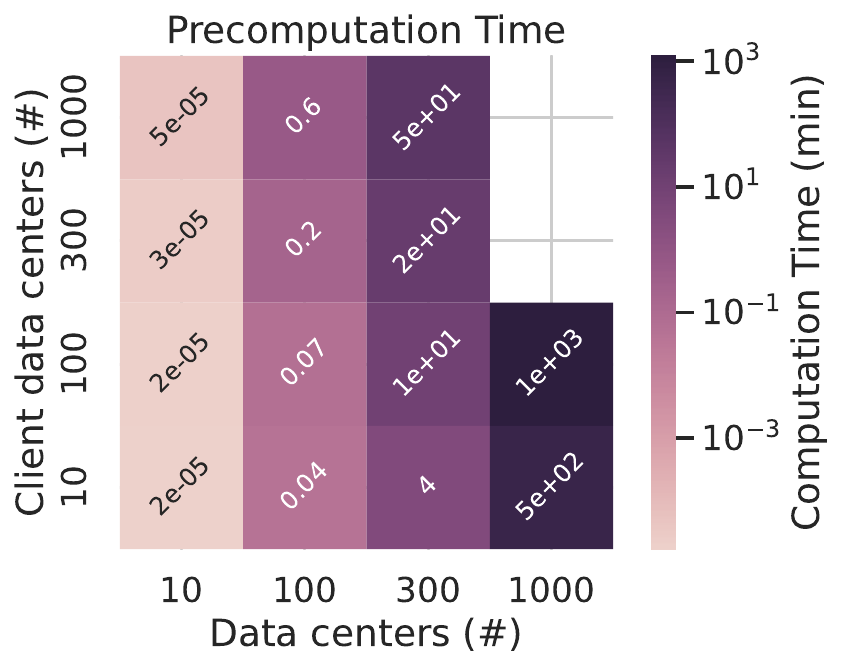}
    \caption{Worst-case single-threaded computation time for \sys{}s when scaling the number of data centers considered for storing objects and the number of client data centers determining the read/write frequency parameter dimensions.}
    \label{fig:construction}
\end{figure}

Figure~\ref{fig:construction} shows the worst-case computation time of Algorithm~\ref{lst:construction} under a range of problem sizes.
Here, we scale the number of data centers that may store objects and the number of read/write frequency parameters of client data centers---scaling the decision space and the cost model dimensions.
We report the computation time for single-threaded computation for the case that all possible data center pairs turn out to be optimal \placements.
We can observe steeply increasing, as can be expected from the nested loops of the algorithm.
Though, the single-threaded computation only takes 20 min at Cloudflare's scale (300 data centers).
Beyond this scale, computation takes $>$7h where filtering of non-optimal placements dominates.

\emph{\textbf{Implications:}}
While seeming brute-force, the offline precomputation of the \sys{} for fault-tolerant object placement is tractable.
Enumerating the quadratic search space is manageable for the number of data centers that cloud vendors have today, and data-parallel computation will further offset overhead.
We point out PQO techniques and further research avenues for precomputation of \sys{}s for more challenging cases in \S\ref{sec:vision}.

\begin{figure*}[ht]
    \begin{minipage}[b]{0.35\linewidth}
            \centering
            \includegraphics[width=\linewidth, trim={0 1.75cm 0 2cm},clip]{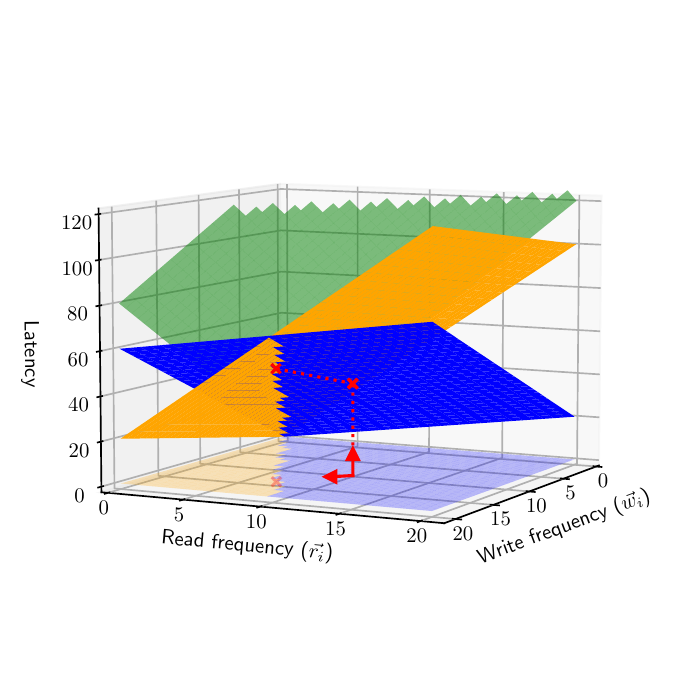}
            \caption{Illustration of vertical and directed ray-shooting onto the hyperplanes that represent the latency of decisions under any read/write frequency.}
            \label{fig:case1:illustration}        
    \end{minipage}
    \hfill
    \begin{minipage}[b]{0.6\linewidth}
        \centering
        \includegraphics[height=4cm]{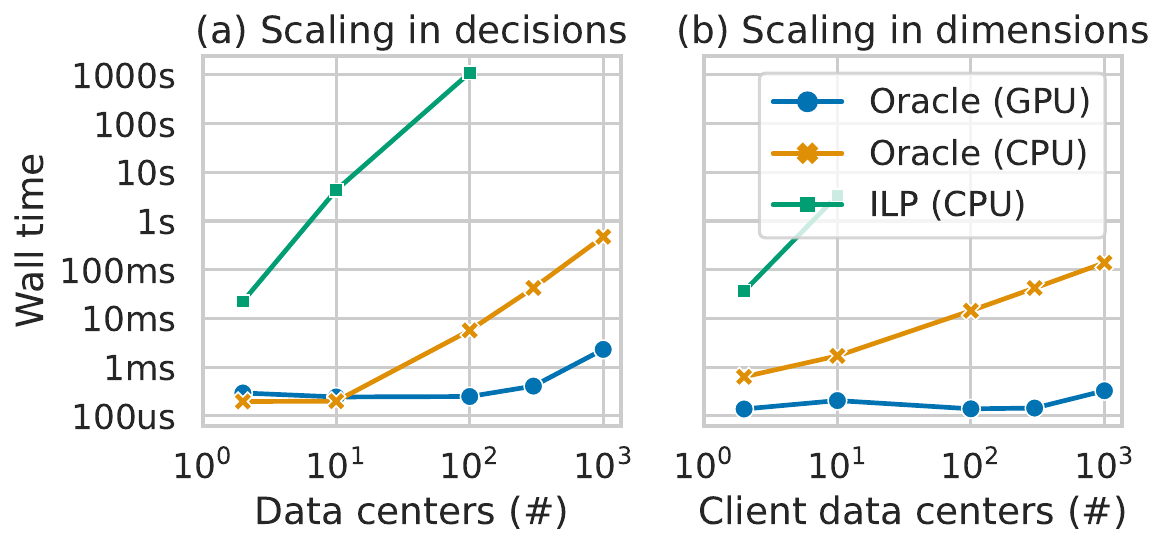}%
        \includegraphics[height=4cm]{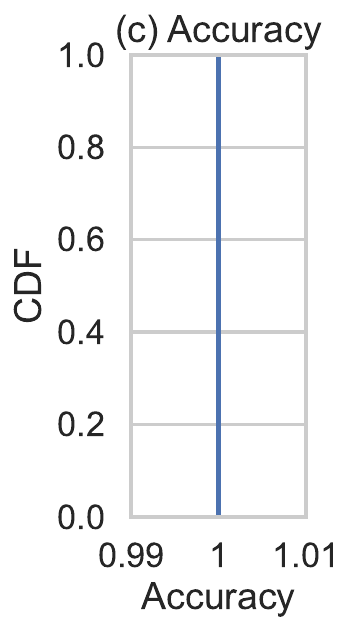}
        \vspace{-1em}
        \caption{Minimization with the \sys{} versus the ILP.
        (a) shows the scalability in the number of decisions in the search space (data centers $D$) under $|C|=300$ and
        (b) in the dimensions of the cost model (client data centers $C$) under $|D|=300$.
        (c) shows the accuracy of the \sys{} relative to the ILP.
        }
        \label{fig:case1:experiment}
    \end{minipage}
    \vspace{-1em}
\end{figure*}

\subsection{Efficient \& Accurate Minimization}
\label{sec:use-cases-minimization}
\textbf{\textit{Which?}}
Consider the processes of deciding which two data centers are the optimal places to store the copies of an object,
as a function of recorded write/read frequencies ($\vec{w}, \vec{r}$) and the network latency $\vec{l}$.
This maps directly to the typical minimization task of Eqs.\ref{eq:argmin}--\ref{eq:constraint}.
Rather than solving the complex ILP of Algorithm~\ref{lst:ILP},
the existence of a \sys{} allows solving this task orders of magnitude faster but equally accurately through simple geometric intersection search---namely \textit{vertical ray-shooting}.

We illustrate vertical ray-shooting on a simplified oracle in Figure~\ref{fig:case1:illustration}.
We consider the ``bounding box'' (i.e., convex polytope) formed by the $parameter \times latency$ hyperplanes of all \placements in the oracle.
Each point on the surface of this polytope is an optimum latency for some vector of input parameters $\vec{w}$ and $\vec{r}$.
Given a specific input, we can start at the ``floor'' of the $parameter \times latency$ space where latency=0 (i.e., $[\vec{w}, \vec{r}, 0]$),
and compute the closest intersection point along a vertical line with the planes above.
This point in $parameter \times latency$ space identifies the plane with the lowest latency for the given parameter---the optimal \placement.

Algorithm~\ref{lst:queryOptimum} implements vertical ray-shooting with a single matrix-vector multiplication.
As shown, vertical ray-shooting is a special case that allows aggressive simplification compared to the general ray-shooting used later.
As such, the minimization task on the \sys{} is very cheap and well suited for GPU-acceleration.
Also, the results are exact by construction---given that the optimal decisions are contained in the \sys{}.

\begin{algorithm}[t]
\begin{algorithmic}[1]
\Procedure{query\_minimum}{$o$, $\vec{a}$}
\LComment{Vertical ray-shooting: argmin on matrix-vector mul.}
\State
$idx,$ $value \gets \argmin_{} (o.\overrightarrow{planes} \cdot \vec{a})$
\State \Return $(o.placements_{idx},\;value)$
\EndProcedure
\end{algorithmic}
\caption{Algorithm for querying the oracle $o$ for the optimal decisions with minimal latency under the given parameter $\vec{a}=[\vec{w},\vec{r}]$.}
\label{lst:queryOptimum}
\end{algorithm}

Figure~\ref{fig:case1:experiment} shows the significant optimization speedup of minimization tasks with the \sys{} compared to the traditional ILP.
In these benchmarks we compare the optimization time of
solving the ILP with Gurobi on a 40-core CPU
and the \sys{} on the CPU as well as a Nvidia V100 GPU\footnote{ILP solvers rely on sequential refinement of solutions that preclude parallelism. Commercial solvers thus operate only on CPUs and do not support GPU acceleration~\cite{gurobi_gpu}.}.

Figure~\ref{fig:case1:experiment}(a) shows the optimization time when scaling the number of data centers $D$ and hence the number of feasible \placements under $|C|=300$---note the log scales.
This affects the complexity of solving the ILP and the size of the \sys{}.
We see the high optimization overhead of the ILP, starting at >10ms for only 2 data centers and reaching >1000s for 100 data centers.
Even for 1000 data centers, the \sys{} remains sub second on the CPU and sub 10 milliseconds on the GPU.
The \sys{} is more than 4 orders of magnitude faster than the ILP.

Figure~\ref{fig:case1:experiment}(b) shows the optimization time when instead scaling the number of client data centers $|C|$, i.e., the number of parameters of the write/read frequencies in the model, under $|D|=300$.
Also under high-dimensional parameters, the \sys{} is orders of magnitude faster.

Finally, Figure~\ref{fig:case1:experiment}(c) shows the accuracy of the \sys{} compared to the ILP for both of the above experiments.
The \sys{} is perfectly accurate---it yields exactly the same optimization results.

\insight{}
Cloud oracles are fast while achieving exact accuracy!
Minimization tasks that use the \sys{} are orders of magnitude faster than ILPs, even if the \sys{} has to cover a large search space, e.g., all 300 Cloudflare data centers,
and the cost functions have hundreds of parameter dimensions.
Using \sys{}s makes accurate minimization affordable at cloud scale and benefits from accelerator hardware.
\vspace{-2em}
\subsection{Efficient Scenario Planning}
\label{sec:what-if}
\textbf{\textit{What If?}}
Consider now a situation in which our online retailer wants to expand from the US to the EU and has to negotiate discounted contracts for reserving capacity in specific EU data centers.
The retailer must explore ``what if'' scenarios to determine how the choice of data centers will change the access patterns and thus affect optimal object placements and access latencies. Such scenario planning is based on ``what if'' analyses, which typically uses stochastic simulation to iteratively explore a vast number of possible scenarios to present anticipated latency profiles.

Simulation techniques like Monte Carlo simulation essentially consist of an outer loop that draws sample parameters from a trace or model, and an inner loop that computes the optimal decision and resulting costs for a given sample~\cite{kroese_why_2014}.
These simulations need to evaluate many samples to confidently judge scenarios. This is infeasible if the inner loop is an ILP that takes seconds or even minutes to compute.
Our \sys{} approach can be used as an alternative lightning-fast inner loop, allowing these simulations to do what they are intended to do: aggressively explore many samples and scenarios.

\begin{algorithm}[t]
\begin{algorithmic}[1]
\Procedure{simulate\_scenario}{$o$, $o'$, $Samples$}
\State $improvements \gets \{\}$ \Comment{Latency improvements of $\vec{a} \in Samples$}
\ForAll {$\vec{a} \in Samples$}
\LComment{Ratio of new optimal latency vs. prior optimal latency}
\State $i \gets min(o'.planes \cdot \vec{a}) / min(o.planes \cdot \vec{a})$
\State $improvements \gets improvements \bigcup \{i\}$
\EndFor
\State \Return $mean(improvements)$, $confidence(improvements)$
\EndProcedure
\end{algorithmic}
\caption{Algorithm for simple Monte Carlo simulation computing the expected latency improvement between two scenarios. $o$ is the \sys{} for the US-only scenario and $o'$ the \sys{} for the expansion to US+EU.}
\label{lst:querySimulation}
\end{algorithm}

Algorithm~\ref{lst:querySimulation} shows that Monte Carlo simulation immediately follows the implementation of the minimization task.
Notably, besides solving for a single sample quickly through vectorization, this simulation also easily scales across several GPUs, simply by replicating the \sys{} matrices and partitioning the samples.

\begin{figure}
    \centering
    \includegraphics[height=4cm]{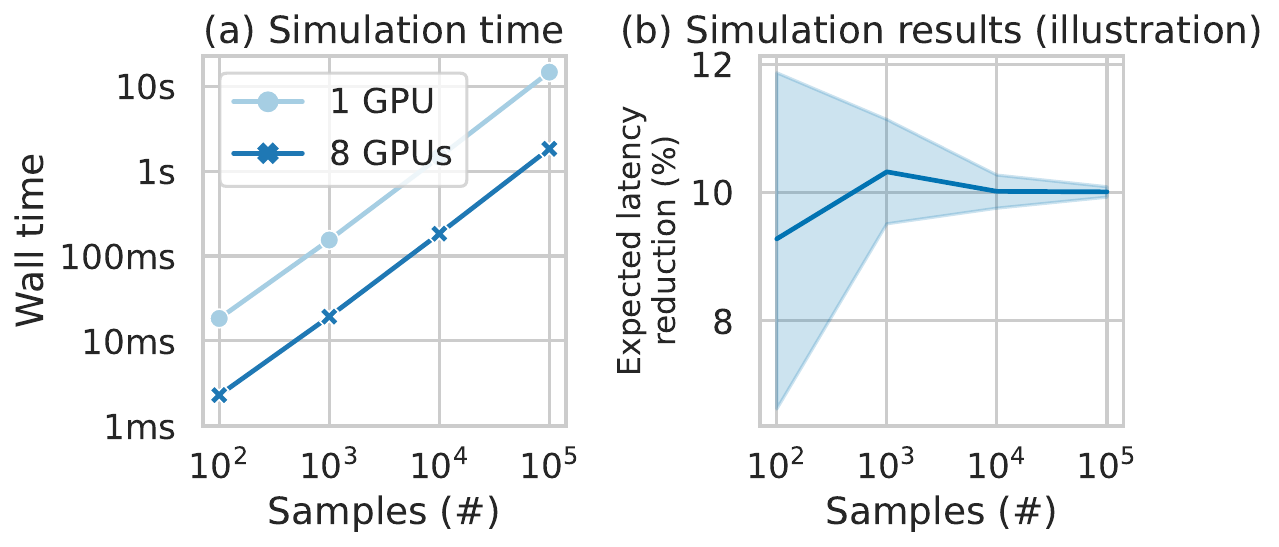}
    \vspace{-1em}
    \caption{Monte Carlo simulation with the \sys{} for 300 data centers ($D$) and client data centers ($C$). (a) Scalability of the \sys{} in the number of samples optimized on 1 GPU and 8 GPUs. (b) An illustration how simulation results (median and 95-confidence interval) converge under the number of samples from a trace.}
    \label{fig:case2:experiment}
\end{figure}

Figure~\ref{fig:case2:experiment} demonstrates the Monte Carlo simulation on top of the \sys{} for judging the latency improvement when expanding from the US to the EU.
We randomly generate access frequency samples and compute the expected latency reduction at the scale of Cloudflare ($|D|=|C|=300$).
Figure~\ref{fig:case2:experiment}(a) shows that the simulation time scales linearly in the number of samples and GPUs.
Figure~\ref{fig:case2:experiment}(b) illustrates that thousands of samples from a trace may be necessary for the median and 95-confidence-interval to converge.
In combination, we can see that the \sys{} on the GPU can provide confident simulation results within milliseconds to seconds---interactive response time.

\insight{}
The \sys{} approach makes stochastic simulation feasible at the scale of the modern clouds. This unlocks a wide range of planning and forecasting scenarios. One interesting family of applications are system configuration tasks like storage placement or other system parameter (``knob'') tuning. A wide range of performance profiles can be created with very low overhead and either presented visually to a human for decision-making, or fed into an automatic configuration service. In either case, decisions can be made to account for a wide variety (or a probability distribution) of ``what if'' scenarios, each one accurately assessed via the \sys{}.
Stochastic methods are not only interesting for applications on top of \sys{}s, but are also interesting avenues for the construction of robust \sys{}s over uncertain cost models, as we point out in \S\ref{sec:vision}.
\subsection{Efficient Migration Planning}
\label{sec:use-cases-drift}

\textbf{\textit{When?}} Finally, consider the task of migration planning for large objects with changing access pattern, e.g., due to diurnal or seasonal popularity in different regions of the globe~\cite{gracia-tinedo_dissecting_2015,taft_p-store_2018}.
Our retailer may want to migrate objects to improve access latency.
However, migrating large objects takes time.
Suppose the online retailer monitors the access pattern and can estimate the migration duration ($x$), knowing the available network bandwidth.
They want to start migrating at $t-x$, where $t$ is the \emph{time when the current placement starts to become suboptimal}. But when is time $t$?

\begin{algorithm}[t]
\begin{algorithmic}[1]
\Procedure{query\_drift\_directed}{$o$, $\vec{a}$, $\vec{d}$}
  \State $curr, cost \gets \Call{query\_minimium}{o, \vec{a}}$
  \State $\overrightarrow{plane} \gets o.\overrightarrow{planes}_{curr}$ \Comment{Plane of current optimum}
  \State $\vec{s} \gets [\vec{a}, cost]$ \Comment{Start point on plane}
  \State $\vec{d'} \gets \vec{d} - ((\vec{d} \cdot \overrightarrow{plane})*\overrightarrow{plane})$ \Comment{Projected drift along plane}
  \State $\overrightarrow{Other} \gets o.\overrightarrow{planes} \setminus \{\overrightarrow{plane}\}$ \Comment{Matrix of remaining planes}
  
  \LComment{Ray-shooting from start $\vec{s}$ in direction of projected drift $\vec{d'}$}
  \State $\overrightarrow{distance} \gets (\overrightarrow{Other} \cdot -\vec{s}) / (\overrightarrow{Other} \cdot \vec{d'})$
  \ForAll{$i \in 0..n$} \Comment{Find closest positive distance}
  \If{$\overrightarrow{distance}_i > 0$ \&\& $\overrightarrow{distance}_i < min\_dist$} 
  \State $min\_idx \gets i, min\_dist \gets \overrightarrow{distance}_i$
  \EndIf
  \EndFor
  \State $\overrightarrow{closest} \gets \vec{s} + (\vec{d'} * min\_dist)$ \Comment{Closest intersection point} 
  \State $next \gets o.decisions_{min\_idx}$ \Comment{Next optimal decision}
  \State $\overrightarrow{param\_next} \gets \overrightarrow{closest}_{0,\dots,n-1}$ \Comment{Next parameter vector}
  \State $cost\_next \gets \overrightarrow{closest}_n$ \Comment{Next latency}

\State \Return $(next, \overrightarrow{param\_next}, cost\_next, min\_dist)$ 
\EndProcedure
\end{algorithmic}
\caption{Algorithm for querying the oracle $o$ for the next optimal decision under the drift $\vec{d}$ of parameter $\vec{a}$.}
\label{lst:drift}
\end{algorithm}

With current approaches, one would have to solve a series of ILPs over quantized time steps to find time $t$.
This exacerbates the already high overhead, forcing current large-scale industrial systems to take a reactive migration approach~\cite{annamalai_sharding_2018}.
In this situation, we can take advantage of the geometric nature of the \sys{} to answer new classes of queries.

\emph{Planning under predictable drift.}
Consider predictable diurnal access pattern. Our retailer may want to migrate objects to evening locations, when people get home and have time for online shopping.
In this situation, the \sys{} can answer the query: 
\emph{How far can the parameters drift in a given direction without changing the optimum?} If we know the direction and rate of change, we can use such a query to tell us precisely when the current optimum will no longer apply; this allows us to plan ahead for that eventuality.

We can directly compute when to migrate under given drift via \emph{directed ray-shooting} on the hyperplanes of the convex polytope representing the \sys{}.
As illustrated in Figure~\ref{fig:case1:illustration}, we can view the problem as a point in $parameter \times latency$ space that moves in the direction and speed of a given parameter drift vector.
The point moves along the hyperplane of the current optimum and eventually collides with a neighboring hyperplane---at $parameter \times latency$ coordinates where the current optimum has the same latency as its neighbor.
This intersection point identifies the next optimal \placementS including the parameters, latency (cost), and time when this next optimum is reached. Due to linearity and convexity, linear intersection search in the direction of the drift suffices to find the intersection point.

Algorithm~\ref{lst:drift} implements the query for directed drift.
It first projects the given drift $\vec{d}$ onto the current optimal hyperplane in $parameter \times latency$ space and then computes the closest intersection with the remaining hyperplanes from the current $parameter \times latency$ point.

\emph{Planning under undirected drift.}

Now consider predictable access peaks.
Before events like Black Friday, our retailer must plan migration of all their objects in advance.
Based on historic information, the retailer can estimate the access frequency of each object and anticipate its optimal placement under some confidence.
One may apply Monte Carlo simulation but in this situation the simulation for each object would require significant computation time, even with the \sys{}.
The \sys{} offers a shortcut via the query: \emph{What is the least parameter drift that changes the optimal configuration?}
This query tells us if migration of an object will be necessary.
If the confidence interval is smaller than the output drift,
then the object will have a single optimal placement.
Otherwise, computationally intensive simulation needs to anticipate a migration plan.

We can efficiently compute closest neighboring optimum in any drift direction via intersection search---after relaxing the problem.
This query can be viewed as a point drifting within the current hyperplane and
we have to find the closest intersection point with the remaining hyperplanes in \emph{any drift direction}.
We are not given a drift vector,
instead we have to compute the vector $\vec{d^*}$ that lies inside the current hyperplane and minimizes the distance from the current point on $x_0$ to some intersection point $x^*$:
\begin{align}
    &min\quad t,\text{ s.t.} &\text{\# minimal distance } t\\
    &x_0 + t \vec{d^*} = x^* \wedge \exists p \in P: x^* \in p&\text{\# intersection with }p_i \\
    &\vec{d^*} \cdot \vec{n} = 0 &\text{\# parallel to current plane}\\
    &||\vec{d^*}||_2 = 1 &\text{\# normalized distance}
\end{align}

Rather than falling back to an expensive solver,
we can utilize the \sys{} when relaxing the problem.
That is, we can utilize \emph{Lagrangian relaxation} to formulate an unconstrained minimization problem for the closest intersection to each neighboring hyperplane individually.
We can efficiently compute this relaxed problem and then simply take the minimum to find the overall closest intersection---the least drift that changes the current optimum.
Appendix~\ref{app:derivation} details our relaxation that 
Algorithm~\ref{lst:drift_conservative} implements to answer the undirected drift query.

\begin{algorithm}[h]
\begin{algorithmic}[1]
\Procedure{query\_drift\_undirected}{$o$: Oracle, $\vec{a}$: Parameter}
\State $curr, cost \gets \Call{query\_minimium}{o, \vec{a}}$
  \State $\overrightarrow{plane} \gets o.\overrightarrow{planes}_{curr}$ \Comment{Plane of current optimum}
  \State $\overrightarrow{Other} \gets o.\overrightarrow{planes} \setminus \{\overrightarrow{plane}\}$ \Comment{Matrix of remaining planes}
  \State $\vec{s} \gets [\vec{a}, cost]$ \Comment{Start point at parameter and cost coordinates}
  \LComment{Solving closest intersection for each plane via relaxation}
  \ForAll{$\vec{p_i} \in \overrightarrow{Other}$}
  \State $\alpha_i \gets \vec{p_i} \cdot \overrightarrow{plane}$
  \State $\vec{v_i} \gets (\alpha_i\overrightarrow{plane} - \vec{p_i}) / \sqrt{1-\alpha_i^2}$
  \State $distance_i \gets (\vec{p_i} \cdot -\vec{s}) / (\vec{p_i} \cdot \vec{v_i})$
  \EndFor
  \State $min\_idx, min\_dist \gets \text{argmin}(\overrightarrow{distance})$
  \State $\overrightarrow{closest} \gets \vec{s} + (\vec{V_{min\_idx}} * min\_dist)$ \Comment{Closest intersection}
  \State $next \gets o.decisions_{min\_idx}$ \Comment{Next optimal decision}
  \State $\overrightarrow{param\_next} \gets \overrightarrow{closest}_{0,\dots,n-1}$ \Comment{Next parameter vector}
  \State $cost\_next \gets \overrightarrow{closest}_n$ \Comment{Next latency}

\State \Return $(next, \overrightarrow{param\_next}, cost\_next, min\_dist)$
\EndProcedure
\end{algorithmic}
\caption{Algorithm for querying the oracle for the next optimal decision under unknown drift of parameter $\vec{a}$.}
\label{lst:drift_conservative}
\end{algorithm}

\begin{figure}[t]
    \centering
    \includegraphics[height=4cm]{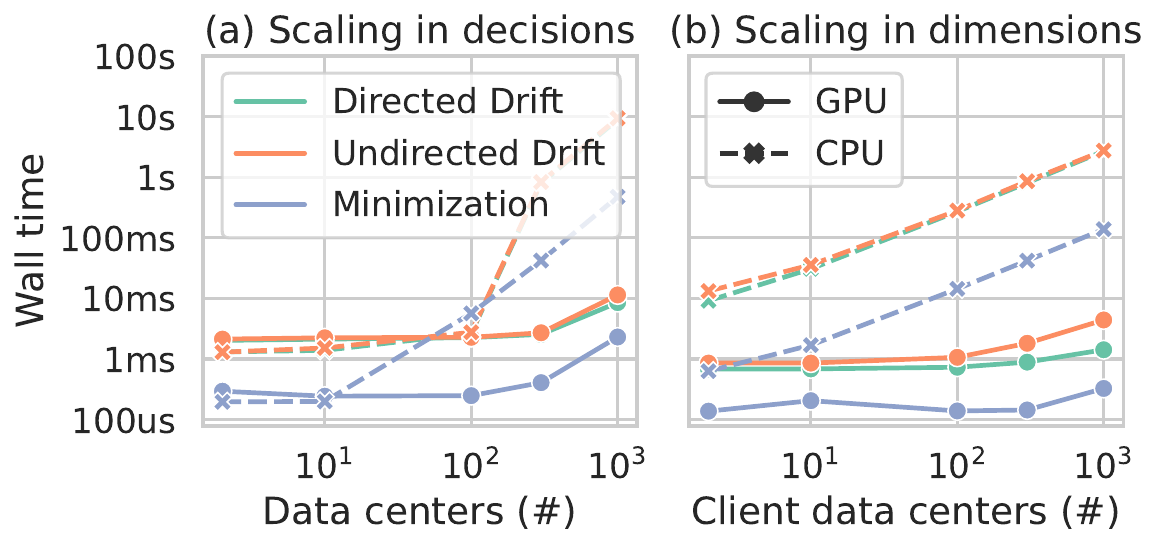}%
    \caption{Scalability of the drift optimization tasks compared to the simple minimization on the \sys{}. (a) shows the scalability in the number of decisions in the search space (data centers $D$) under $|C|=300$ and (b) in the dimensions of the cost model (client data centers $C$) under $|D|=300$.
    }
    \label{fig:case3:experiment}
\end{figure}

Figure~\ref{fig:case3:experiment} shows the optimization time of the advanced drift queries in comparison to the simple minimization task.
Here, we repeat the scaling experiment from earlier.
We can see that even for these advanced optimization tasks the optimization time remains sub 10 milliseconds on the GPU and overall very manageable compared to the ILP solving the simple minimization.

\insight{}
The low-overhead drift queries highlight the power of geometric computations on the \sys{} and indicate interesting applications in dynamic optimization.
Not only do drift queries allow us to plan ahead for a given drift,
but they may also benefit time series-based dynamic optimization.
A common approach is to predict a time series of workload parameters and then optimize over quantized time slots~\cite{westenbroek_stability_2021,rawlings_model_2017,goel_thinking_2017,mansouri_cost_2019, chen_using_2016}.
Fine-grained time slots allow for high accuracy but suffer from high overhead while
coarse-grained time slots minimize overhead but also worsen accuracy.
Instead, traveling the surface of the time series and the cost functions has the potential to offer high accuracy and avoid high overhead (see annealing methods described in \S\ref{sec:vision:ml}).
\section{Toward Cloud Oracles}
\label{sec:vision}

We have one major question left to address: For what kind of problems can we compute a \sys{}?
Akin to PQO, offline optimization for a \sys{} requires enumeration of candidate decisions and subsequent selection of optimal decisions based on a cost model~\cite{doshi_kepler_2023,vaidya_leveraging_2021}.
The three conditions enabling efficient offline optimization are: 1) a relatively fixed cost model; 2) a mostly linear cost model; 3) independence of decisions.

The placement problem for Durable Objects fulfills these three conditions and has a small search space.
This enables straightforward enumeration and selection of decisions for the \sys{}.
We have seen that the resulting \sys{} offers large speedups and improved accuracy for a range of important optimization tasks.
Indeed, this beneficial scenario applies to placement problems that do not require replication and have elastic capacity.
For example, cost optimization of S3 objects has a linear pay-per-use pricing model and a search space of 99 AZs with each 6 storage tiers (594 decisions)~\cite{aws_s3_tiers, aws_datacenters}.

In contrast, ``bin-packing-style'' problems violate condition 3, hence making offline optimization extremely difficult.
For example, consider data placement to minimize average access latency for a DBMS that runs on a VM with local and remote storage.
Local storage has low access latency but limited capacity, so
we must decide which tables to pack in there.
For optimal placement, we must jointly decide for all tables at once, which blows up the problem with little chance of pruning.
Compared to independent decisions, this bin-packing problem appears combinatorially intractable to precompute offline.

Along the following three research questions, we now discuss problems that challenge the application of \sys{}s but offer promising directions:
\begin{enumerate}[noitemsep,nosep,left=5pt .. 17pt]
    \item How to make enumeration of candidate decisions tractable in large configuration spaces of cloud-scale systems?
    \item How to select optimal decisions under non-linear cost models?
    \item How to select optimal decisions under variance in cost models, e.g., network latency?
\end{enumerate}

\vspace{-1em}
\subsection{Enumeration in Large Search Spaces}
\label{sec:vision:pqo}

Placement \emph{with replication} still fits our three conditions, but optimizing for combinations of placements blows up the search space.
For example, consider geo-replication of objects for latency minimization or geo-replication of DBMSs for fault-tolerance.
These problems challenge efficient candidate enumeration.
Facing a massive search space, one important lesson to learn from (parametric) query optimization is that effective pruning is imperative~\cite{leis_query_2018,trummer_multi-objective_2017,bruno_configuration-parametric_2008}.
No matter how efficient the selection, excessive enumeration of candidates will prohibit tractable offline optimization.
Pruning is critical, but it is heuristic.
The challenge is that pruning must exploit problem structure,
making algorithms hard to design and their performance case-sensitive.

An interesting direction for effective pruning in the cloud is algorithmic meta-optimization~\cite{comden_online_2019,gottschlich_three_2018,yi_automated_2023}.
Given the scale and diversity of the cloud, manual design and selection of pruning algorithms is challenging.
Instead, meta-optimization of pruning algorithms could be an approach for automatic classification of problem structure and algorithm specialization.
Algorithmic meta-optimization is an ongoing research effort and may benefit from cost estimation and other techniques à la query optimization.

\subsection{Decisions Selection under Non-linear Costs}
\label{sec:vision:ml}

Life is not always linear, and neither are cost models.
For example, state-of-the-art cost models for runtime estimation are learned models.
For application to these problems the computation of \sys{}s has to support learned cost models.
These cost models precisely capture complex relationships between workload parameters and costs,
as shown for cardinality and cost estimation~\cite{ready_for_cardinality,marcus_bao_2021,hilprecht_zero-shot_2022,hilprecht_deepdb_2020}.
However, this complicates the navigation of their cost surface, required for selecting the optimal decisions for the \sys{}.

Gradient descent, annealing, and differential evolution~\cite{onwubolu_new_2004,storn_differential_1997,ioannidis_parametric_1997} offer directions for computing \sys{}s on models with complex cost surfaces.
Convex linear models allow efficient computation in \sys{}s, as their smooth surfaces obviate the optimal lowest-cost decisions.
Learned models have rough surfaces with hills that obfuscate lowest-cost decisions.
A simple idea to compute optimal decisions under complex cost models is the pair-wise comparison of candidates, similar to the multi-objective PQO approach~\cite{trummer_multi-objective_2017}.
That is, the cost of two candidates are the learned functions $f^{*}(\vec{a})$ and $g^{*}(\vec{a})$.
The cost difference between the two candidates is the function $h(\vec{a})=f^{*}(\vec{a}) - g^{*}(\vec{a})$.
If $h$ has a negative value, there exists a parameter for which $f^{*}$ is cheaper than $g^{*}$ otherwise $f^{*}$ is dominated and never optimal.
We can use annealing methods---that Ioannidis et al.~\cite{ioannidis_parametric_1997} originally proposed for PQO---to search the surface of $h$ for negative values.
Afterwards, we can construct a \sys{} with a compact model as lookup structure for selected optimal decisions.
Research on robust knowledge distillation of a compact ``student model'' for given decisions is required, see~\cite{allen-zhu_towards_2022}.
Also, efficient methods for computing optimal decisions are required for problems with large parameter and search spaces.

Note that annealing techniques are similarly interesting for solving complex optimization tasks on the surface of the cost model, like our drift query.

\vspace{-1em}
\subsection{Decisions Selection under Uncertainty}

Random variables that capture uncertainty challenge the static condition (1) for offline optimization.
For example, placement of objects onto edge locations may involve significant variation in network latency.
Also, monitoring of workload parameters like access frequency in a large distributed systems will not be exact, but may involve uncertainty due to sampling.
Particular challenges are how to derive a probabilistic cost model with random cost coefficients and
how to select optimal decisions offline in the presence of random cost coefficients and random workload parameters.

Interesting directions are learned cost models capturing uncertainty and robust parametric optimization.
For example, one may explore the transfer of learned cardinality estimation with uncertainty~\cite{liu_fauce_2021}.
A starting point for robust optimization is the work~\cite{lei_robust_2014}, which considers robust load distribution for distributed streaming systems.
Offloading robust optimization to the precomputation of \sys{}s not only leads to fast and robust online optimization for cloud systems, but may also provide opportunities for offline verification.
The behavior of online optimization under uncertainty as well as pure ML-based optimization is difficult to foresee.
The materialized decisions in a \sys{} instead are known explicitly, thus enabling more stringent offline verification.
\vspace{-1em}
\section*{Conclusion}
The cloud is both powerful and complex. To harness its power to the fullest, we need to make accurate and rapid choices in enormous decision spaces. Classical ILP does this accurately, but slowly. Machine learning approaches trade accuracy for speed,while often introducing ongoing costs of ML ops. We have proposed a best-of-both-worlds design, \sys{}s. This approach follows the ML script of slow offline training combined with fast online inference but does so with the data-independent explainable solutions of ILP.


\newcommand{\shownote}[1]{\unskip} 
\bibliographystyle{ACM-Reference-Format}
\bibliography{skycache-zotero,skycache-local,conor}

\appendix
\begin{appendices}

\section{Undirected drift: Lagrangian Relaxation} \label{app:derivation}

\setcounter{equation}{0}
\renewcommand\theequation{A.\arabic{equation}}

We present a derivation to compute the closest intersection point from a given parameter along the current optimal plane, as utilized in Section~\ref{sec:use-cases-drift}. \Cref{fig:planes} illustrates two hyperplanes in \(\mathbb{R}^n\), denoted as \(P_0\) and \(P_i\), each defined by their respective unit normal vectors \(\vect{n}_0\) and \(\vect{n}_i\). Consider a point \(\vect{x}_0\) on \(P_0\) and a unit vector \(\vect{v}_0\) originating from \(\vect{x}_0\) and lying tangentially to \(P_0\). Importantly, \(\vect{v}_0\) is perpendicular to \(\vect{n}_0\). Our goal is to determine the intersection point \(\vect{x}_i\) on \(P_i\), where a ray in the direction of \(\vect{v}_0\) intersects \(P_i\).

The intersection point \(\vect{x}_i\) can be expressed as
\begin{equation}
    \vect{x}_i = \vect{x}_0 + t \vect{v}_0, \label{eq:xi}
\end{equation}
where \(t\) is a scalar. To find \(t\), we consider an arbitrary point \(\vect{o}\) on \(P_i\) and the orthogonality of \(\vect{n}_i\) with \(P_i\), giving \((\vect{x}_i - \vect{o}) \cdot \vect{n}_i = 0\). Assuming \(P_i\) passes through the origin, we can simplify the formulation by setting \(\vect{o}\) to the origin, leading to \(\vect{x}_i \cdot \vect{n}_i = 0\). Taking the dot product of equation \eqref{eq:xi} with \(\vect{n}_i\), and solving for \(t\) results in
\begin{equation}
    t = - \frac{\vect{x}_0 \cdot \vect{n}_i}{\vect{v}_0 \cdot \vect{n}_i}. \label{eq:t}
\end{equation}

Our interest lies in identifying an optimal direction \(\vect{v}_0\) that minimizes the distance between \(\vect{x}_i\) and \(\vect{x}_0\).

\begin{figure}[t!]
    \centering
    \begin{asy}[width=0.6\linewidth]
        include "figures/deviationOfOptimalPoints.asy";
    \end{asy}
    \vspace{-1.5em}
    \caption{Schematic representation of the ray shooting method from point \(\vect{x}_0\) in direction \(\vect{v}_0\) on hyperplane \(P_0\) to locate point \(\vect{x}_i\) on hyperplane \(P_i\).}
    \label{fig:planes}
    \vspace{-1em}
\end{figure}

\begin{proposition}
    Suppose \(\vect{x}_0 \in P_0\) is fixed and the unit vector \(\vect{v}_0\) can vary while passing through \(\vect{x}_0\) and remaining tangent to \(P_0\). Under the condition that \(P_0\) and \(P_i\) are non-parallel, the distance \(\Vert \vect{x}_i - \vect{x}_0 \Vert_2\) is minimized if
    \begin{equation}
        \vect{v}_0 = \pm \frac{\alpha \vect{n}_0 - \vect{n}_i}{\sqrt{1 - \alpha^2}}, \label{eq:v0}
    \end{equation}
    where \(\alpha \coloneqq \vect{n}_i \cdot \vect{n}_0\).
\end{proposition}

\begin{proof}
    From equation \eqref{eq:xi} and the unit norm of \(\vect{v}_0\), we get \(\Vert \vect{x}_i - \vect{x}_0 \Vert_2 = \vert t \vert \). Minimizing \(\vert t \vert\) involves maximizing the absolute value of the denominator in equation \eqref{eq:t}, as the numerator remains constant. The optimization problem is thus defined as
    \begin{equation*}
        \max_{\vect{v}_0 \in \mathbb{R}^n} \; \vert \vect{v}_0 \cdot \vect{n}_i \vert,
    \end{equation*}
    subject to the constraints
    \begin{equation*}
        \vect{v}_0 \cdot \vect{v}_0 = 1 \quad \text{and} \quad \vect{v}_0 \cdot \vect{n}_0 = 0,
    \end{equation*}
    which ensure that \(\vect{v}_0\) remains a unit vector tangent to \(P_0\). Using the method of Lagrange multipliers, we define the Lagrangian function as
    \begin{equation*}
        \mathcal{L}(\vect{v}_0, \lambda, \mu) \coloneqq \vert \vect{v}_0 \cdot \vect{n}_i \vert + \lambda (\vect{v}_0 \cdot \vect{v}_0 - 1) + \mu \vect{v}_0 \cdot \vect{n}_0,
    \end{equation*}
    with \(\lambda\) and \(\mu\) as Lagrangian multipliers. The optimal solution is found at the stationary point of \(\mathcal{L}\), where its partial derivatives vanish. Namely,
    \begin{subequations}
    \begin{align}
        \frac{\partial \mathcal{L}}{\partial \vect{v}_0} &= s \vect{n}_i + 2 \lambda \vect{v}_0 + \mu \vect{n}_0 = \vect{0}, \label{eq:L1} \\
        \frac{\partial \mathcal{L}}{\partial \lambda} &= \vect{v}_0 \cdot \vect{v}_0 - 1 = 0, \label{eq:L2} \\
        \frac{\partial \mathcal{L}}{\partial \mu} &= \vect{v}_0 \cdot \vect{n}_0 = 0, \label{eq:L3}
    \end{align}
    \end{subequations}
    where \(s \coloneqq \sgn (\vect{v}_0 \cdot \vect{n}_i)\).

    To determine \(\mu\), we take the dot product of equation \eqref{eq:L1} with \(\vect{n}_0\). Incorporating the constraint from equation \eqref{eq:L3} and recognizing that \(\vect{n}_0 \cdot \vect{n}_0 = 1\), we deduce that \(\mu = - s\vect{n}_i \cdot \vect{n}_0 = -s\alpha\). Consequently, resolving equation \eqref{eq:L1} for \(\vect{v}_0\) leads us to \(\vect{v}_0 = s(\alpha \vect{n}_0 - \vect{n}_i) / (2 \lambda)\). The next step involves the determination of \(\lambda\).

    We proceed by taking the dot product of equation \eqref{eq:L1} with \(\vect{v}_0\) and subsequently applying equations \eqref{eq:L2} and \eqref{eq:L3}, which yields \(2 \lambda = -s\vect{n}_i \cdot \vect{v}_0\). On the other hand, we also take the dot product of equation \eqref{eq:L1} with \(s\vect{n}_i\), considering the fact that \(\vect{n}_i \cdot \vect{n}_i = 1\) and \(s^2 = 1\). This operation results in the equation \(1 + 2 \lambda (s \vect{n}_i \cdot \vect{v}_0) - \alpha^2 = 0\). By eliminating \(s \vect{n}_i \cdot \vect{v}_0\) from these two equations and solving for \(2\lambda\), we find that \(2 \lambda = \pm \sqrt{1 - \alpha^2}\). Hence, we derive \(\vect{v}_0 = \pm s (\alpha \vect{n}_0 - \vect{n}_i) / \sqrt{1 - \alpha^2}\). Given the arbitrary nature of the sign in this expression, we can equate \(\pm s\) to \(\pm 1\), thereby completing the proof.
\end{proof}

\begin{remark}
    When \(P_0\) and \(P_i\) are parallel, equation \eqref{eq:v0} becomes undefined, as \(\alpha = \pm 1\). In this case, \(t = 0\), leading to the trivial solution \(\vect{x}_i = \vect{x}_0\).
\end{remark}

\end{appendices}

\end{document}